\documentclass[journal,10pt]{IEEEtran}
\usepackage{times}
\usepackage{cite}
\usepackage{color}
\usepackage{epsf}
\usepackage{epsfig}
\usepackage{graphicx}
\usepackage{graphics}
\usepackage{epstopdf}
\usepackage[footnotesize]{caption}
\usepackage{amsmath}
\usepackage{amssymb}
\usepackage{amsthm}
\usepackage{amsxtra}
\usepackage[ruled]{algorithm2e}
\usepackage{algorithmic}
\usepackage{enumerate}
\usepackage{multirow}
\usepackage{enumerate}
\usepackage{amssymb}
\usepackage{lipsum}
\usepackage{setspace}
\usepackage{multirow}
\usepackage{fancyhdr}
\newtheorem{corollary}{Corollary}
\newtheorem{theorem}{\bf Theorem}

\newtheorem{definition}{\bf Definition}

\usepackage{rotating} 
\usepackage{graphicx} 

\newcommand{\Rmnum}[1]{\expandafter\@slowromancap\romannumeral #1@}

\makeatother
\begin{document}
\title{Energy Storage Sharing in Smart Grid: A Modified Auction Based Approach}
\author{Wayes~Tushar,~\IEEEmembership{Member,~IEEE,}~Bo~Chai,~Chau~Yuen,~\IEEEmembership{Senior Member,~IEEE,}
 Shisheng Huang,~\IEEEmembership{Member,~IEEE,}
David B. Smith,~\IEEEmembership{Member, IEEE,} H.~Vincent~Poor~\IEEEmembership{Fellow,~IEEE,} and Zaiyue Yang,~\IEEEmembership{Member, IEEE}
\thanks{W. Tushar and C. Yuen are with Singapore University of Technology and Design (SUTD), 8 Somapah Road, Singapore 487372 (Email: \{wayes\_tushar, yuenchau\}@sutd.edu.sg).}
\thanks{B. Chai is with the State Grid Smart Grid Research Institute, Beijing, 102211, China (Email: chaibozju@gmail.com).}
\thanks{S. Huang is with the Ministry of Home Affairs, Singapore (Email: shisheng@gmail.com).}
\thanks{D. B. Smith is with the National ICT Australia (NICTA), ACT 2601, Australia and adjunct with the Australian National University (Email: david.smith@nicta.com.au).}
\thanks{H. V. Poor is with the School of Engineering and Applied Science at Princeton University, Princeton, NJ, USA (Email: poor@princeton.edu).}
\thanks{Z. Yang is with the State Key Laboratory of Industrial Control Technology at Zhejiang University, Hangzhou, China (Email: yangzy@zju.edu.cn).}
\thanks{This work is supported in part by the Singapore University of Technology and Design (SUTD) through the Energy Innovation Research Program (EIRP) Singapore NRF2012EWT-EIRP002-045 and IDC grant IDG31500106, and in part by the U.S. National Science Foundation under Grant ECCS-1549881. D. B. Smith's work is supported by NICTA, which is funded by the Australian Government through the Department of Communications and the Australian Research Council.}
}
\IEEEoverridecommandlockouts
\maketitle
\begin{abstract}
This paper studies the solution of joint energy storage (ES) ownership sharing between multiple shared facility controllers (SFCs) and those dwelling in a residential community. The main objective is to enable the residential units (RUs) to decide on the fraction of their ES capacity that they want to share with the SFCs of the community in order to assist them storing electricity, e.g., for fulfilling the demand of various shared facilities. To this end, a \emph{modified} auction-based mechanism is designed that captures the interaction between the SFCs and the RUs so as to determine the auction price and the allocation of ES shared by the RUs that governs the proposed joint ES ownership. The fraction of the capacity of the storage that each RU decides to put into the market to share with the SFCs and the auction price are determined by a noncooperative Stackelberg game formulated between the RUs and the auctioneer. It is shown that the proposed auction possesses the \emph{incentive compatibility} and the \emph{individual rationality} properties, which are leveraged via the unique Stackelberg equilibrium (SE) solution of the game. Numerical experiments are provided to confirm the effectiveness of the proposed scheme.
\end{abstract}
\begin{IEEEkeywords}
Smart grid, shared energy storage, auction theory, Stackelberg equilibrium, strategy-proof, incentive compatibility.
\end{IEEEkeywords}
 \setcounter{page}{1}
\section{Introduction}\label{sec:introduction}
\IEEEPARstart{E}{nergy} storage (ES) devices are expected to play a significant role in the future smart grid due to their capabilities of giving more flexibility and balance to the grid by providing a back-up to the renewable energy~\cite{Silvestre-TII:2014,Garcia-TII:2014,Fang-J-CST:2012,Liu-STSP:2014,YiLiu-TIE:2015,Naveed-TSG:2015,Naveed-Elsevier:2016,Tushar-ITS:2015,Shisheng-Elsevier:2015}. ES can improve the electricity management in a distribution network, reduce the electricity cost through opportunistic demand response, and improve the efficient use of energy~\cite{Wang-JTSG:2013}. The distinct features of ES make it a perfect candidate to assist in residential demand response by altering the electricity demand due to the changes in the balance between supply and demand. Particularly, in a residential community setting, where each household is equipped with an ES, the use of ES devices can significantly leverage the efficient flows of energy within the community in terms of reducing cost, decarbonization of the electricity grid, and enabling effective demand response (DR).

However, energy storage requires space. In particular for large consumers like shared facility controllers (SFCs) of large apartment buildings~\cite{Tushar-TIE:2014}, the energy requirements are very high, which consequently necessitates the actual installment of very large energy storage capacity. The investment cost of such storage can be substantial whereas due to the random usage of the facilities (depending on the usage pattern of different residents) some of the storage may remain unused. Furthermore, the use of ESs for RUs is very limited for two reasons \cite{Wang-JTSG:2013}: firstly, the installation cost of ES devices is very high and the costs are entirely borne by the users. Secondly, the ESs are mainly used to save electricity costs for the RUs rather than offer any support to the local energy authorities, which further makes their use economically unattractive. Hence, there is a need for solutions that will capture both the problems related to space and cost constraints of storage for SFCs and the benefit to RUs for supporting third parties.

To this end, numerous recent studies have focused on energy management systems with ES devices as we will see in the next section. However, most of these studies overlook the potential benefits that local energy authorities such as SFCs can attain by jointly sharing the ES devices belonging to the RUs. Particularly due to recent cost reduction of small-scale ES devices, sharing of ES devices installed in the RUs by the SFCs has the potential to benefit both the SFCs and the RUs of the community as we will see later. In this context, we propose a scheme that enables joint ES ownership in smart grid. During the sharing, each RU leases the SFCs a fraction of its ES device to use, and charges and discharges from the rest of its ES capacity for its own purposes. On the contrary, each SFC exclusively uses its portion of ES devices leased from the RUs. This work is motivated by \cite{Wang-JTSG:2013}, in which the authors discussed the idea of joint ownership of ES devices between domestic customers and local network operators, and demonstrated the potential system-wide benefits that can be obtained through such sharing. However, no policy has been developed in \cite{Wang-JTSG:2013} to determine how the fraction of battery capacity, which is shared by the network operators and the domestic users, is decided.

Note that, as an owner of an ES device, each RU can decide whether or not to take part in the joint ownership scheme with the SFCs and what fraction of the ES can be shared with the SFCs. Hence, there is a need for solutions that can capture this decision making process of the RUs by interacting with the SFCs of the network. In this context, we propose a joint ES ownership scheme in which by participating in storage sharing with the SFCs, both the RUs and SFCs benefit economically. Due to the interactive nature of the problem, we are motivated to use auction theory to study this problem~\cite{klemperer-JES:1999}.

Exploiting the two-way communications aspects, auction mechanisms can exchange information between users and electricity providers, meet users' demands at a lower cost, and thus contribute to the economic and environmental benefits of smart grid\footnote{Please note that such a technique can be applied in the real distribution network such as in electric vehicle charging stations by using the two-way information and power flow infrastructure of smart grids~\cite{Fang-J-CST:2012}.}~\cite{Ma-JTSG:2014}. In particular, 1) we modify the \emph{Vickrey} auction technique~\cite{Vickrey-JF:1961} by integrating a Stackelberg game between the auctioneer and the RUs and show that the modified scheme leads to a desirable joint ES ownership solution for the RUs and the SFCs. To do this, we modify the auction price derived from the Vickrey auction, to benefit the owner of the ES, through the adaptation of the adopted game as well as keep the cost savings to the SFCs at the maximum; 2)  We study the attributes of the technique, and show that the proposed auction scheme possesses both the \emph{incentive compatibility} and the \emph{individual rationality} properties leveraged by the unique equilibrium solution of the game; 3) We propose an algorithm for the Stackelberg game that can be executed distributedly by the RUs and the auctioneer, and the algorithm is shown to be guaranteed to reach the desired solution. We also discuss how the proposed scheme can be extended to the time varying case; and 4) Finally, we provide numerical examples to show the effectiveness of the proposed scheme.

The importance and necessity of the proposed study with respect to actual operation of smart grid lies in assisting the SFCs of large apartment buildings in smart communities to reduce space requirements and investment costs of large energy storage units. Furthermore, by participating in storage sharing with the SFCs, the RUs can benefit economically, which can consequently influence them to efficiently schedule their appliances and thus reduce the excess use of electricity. We stress that multi-agent energy management schemes are not new in the smart grid paradigm and have been discussed in \cite{Tushar-TIE:2014,chaibo-TSG:2014} and \cite{Maharjan-JTSG:2013}. However, the scheme discussed in the paper differs from these existing approaches in terms of the considered system model, chosen methodology and analysis, and the use of the set of rules to reach the desired solution.

The remainder of the paper is organized as follows. We provide a comprehensive literature review of the related work in Section~\ref{sec:literature} followed by the considered system model in Section~\ref{sec:system-model}. Our proposed modified auction-based mechanism is demonstrated in Section~\ref{sec:auction-ownership} where we also discuss how the scheme can be adopted in a time varying environment. The numerical case studies are discussed in Section~\ref{sec:case-study}, and finally we draw some concluding remarks in Section~\ref{sec:conclusion}.

\section{State-of-The Art}\label{sec:literature}
In the recent years, there has been an extensive research effort to understand the potential of ES devices for residential energy management \cite{Siano:2014}. This is mainly due to their capabilities in reducing the intermittency of renewable energy generation~\cite{Denholm:2010} as well as lowering the cost of electricity~\cite{Cao:2004}. The related studies can be divided into two general categories. The first category of studies consisting of \cite{Sechilariu-JTIE:2013,Carpinelli-TSG:2013}, which assume that the ESs are installed within each RU premises and are used solely by the owners in order to perform different energy management tasks such as optimal placement, sizing and control of charging and discharging of storage devices.

The second type of studies deal with ES devices that are not installed within the RUs but located in a different location such as in electric vehicles (EVs). Here, the ESs of EVs are used to provide ancillary services for RUs~\cite{GookKim:2013,VanRoy:2014,RongYu:2014} and local energy providers~\cite{JunhaoLin:2014,JunTan:2014,Igualada:2014}. Furthermore, another important impact of ES devices on residential distribution grids is studied in \cite{Geth-PESGM:2010} and \cite{Nukamp-JTSG:2013}. In particular, these studies focus on how the use of ES devices can bring benefits for the stakeholders in external energy markets. In \cite{Geth-PESGM:2010}, the authors propose a multi-objective optimization method for siting and sizing of ESs of a distribution grid to capture the trade-offs between the storage stakeholders and the distribution system operators. Furthermore, in \cite{Nukamp-JTSG:2013}, optimal storage profiles for different stakeholders such as distribution grid operators and energy traders are derived based on case studies with real data. Studies of other aspects of smart grid can be found in \cite{Tushar-TSG:Dec2015,WT-Li-Access:Nov2015,AqsaNaeem-Access:Nov2015,YLiu-TSG:June2015,Xiumin:Elsevier-July2015,Zhang-ACMN:2015,HengZhang-TCST:2015}.

As can be seen from the above discussion, the use of ES devices in smart grid is not only limited to address the intermittency of renewable generation ~\cite{Denholm:2010} and assisting users to take part in energy management to reduce their cost of electricity~\cite{Cao:2004,Carpinelli-TSG:2013} but also extends to assisting the grid (or, other similar energy entities such as an SFC)~\cite{ZhenpoWang:2013} and generating revenues for stakeholders~\cite{Geth-PESGM:2010,Nukamp-JTSG:2013}. However, one similarity between most of the above mentioned literature is that only one entity owns the ES and uses it according to its requirements. Nonetheless, this might not always be the case if there are large number of RUs\footnote{Each RU may participate as a single entity or as a group where RUs connected via an aggregator~\cite{Wayes-J-TSG:2012}.} in a community. In this regard, considering the potential benefits of ES sharing, as discussed in \cite{Wang-JTSG:2013}, this paper investigates the case in which the SFCs in a smart community are allowed to share some fraction of the ESs owned by the RUs through a third party such as an auctioneer or a community representative.

The proposed modified auction scheme differs from the existing techniques for multi-agent energy management such as those in \cite{Tushar-TIE:2014,chaibo-TSG:2014,Maharjan-JTSG:2013} in a number of ways. Particularly, in contrast to these studies, the proposed auction scheme captures the interaction between the SFCs and the RUs, whereby the decision on the auction price is determined via a Stackelberg game. By exploiting auction rules including the determination rule, payment rule, and allocation rule the interaction between the SFCs and RUs is greatly simplified. For instance, the determination rule can easily identify the number of RUs that are participating in the auction process, which further leverage the determination of the auction price via the Stackelberg game in the payment rule. Furthermore, on the one hand the work here complements the existing works focusing on the potential of ES for energy management in smart grid. On the other hand, the proposed work has the potential to open new research opportunities in terms of  control of energy dispatch from ES, the size of ES, and exploring other interactive techniques such as cooperative games and bi-level optimization for ES sharing. 
\section{System Model}\label{sec:system-model}
Let us consider a smart community that consists of a large number of RUs. Each RU can be an individual home, a single unit of a large apartment complex, or a large number of units connected via an aggregator that acts as a single entity~\cite{Gkatzikis:2013,Wayes-J-TSG:2012,Tushar-TSG:2013}. Each RU is equipped with an ES device that the RU can use to store electricity from the main grid or its renewable energy sources, if there are any, or can perform DR management according to the real-time price offered by the grid. The ES device can be a storage device installed within each RU premises or can be the ES used for the RU's electric vehicles. The entire community is considered to be divided into a number of blocks, where each block consists of a number of RUs and an SFC. Each SFC $m\in\mathcal{M}$, where $\mathcal{M}$ is the set of all SFCs and $M=|\mathcal{M}|$, is responsible for controlling the electrical equipment and machines such as lifts, parking lot lights and gates, water pumps, and lights in the corridor area of a particular block of the community, which are shared and used by the residents of that block on regular basis. Each SFC is assumed to have its own renewable energy generation and is also connected to the main electricity grid with appropriate communication protocols.
\begin{figure}[t!]
\centering
\includegraphics[width=0.9\columnwidth]{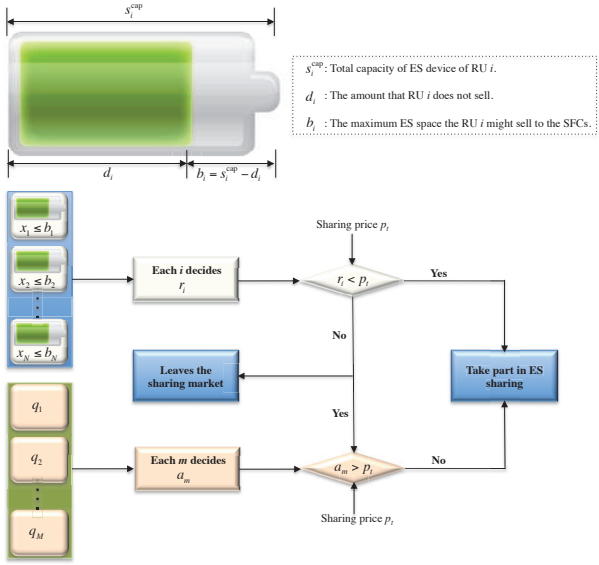}
\caption{The fraction of the ES capacity that an RU $i$ is willing to share with the SFCs of the community.} \label{fig:batteryshare}
\end{figure}

Considering the fact that the nature of energy generation and consumption is highly sporadic~\cite{Saad-CSmartgridComm:2011}, let us assume that the SFCs in the community need some extra ESs to store their electricity after meeting the demand of their respected shared facilities at a particular time of the day. This can be either due to the fact that some SFCs do not have their own ESs~\cite{Tushar-TIE:2014} or that the ESs of the SFCs are not large enough to store all the excess energy at that time. It is important to note that the ES requirement of the SFCs can stem from any type of intermittent generation profile that the SFCs or RUs can adopt. For example, one can consider that the proposed scheme is based on a hybrid generation profile comprising both solar and wind generation. However, the proposed technique is equally suitable for other types of intermittent generation as well. We assume that there are $N = |\mathcal{N}|$ RUs, where $\mathcal{N}$ is the set of all RUs in the system, that are willing to share some parts of their ES with the SFCs of the network. The battery capacity of each RU $i\in\mathcal{N}$ is $s_i^\text{cap}$, and each RU $i$ wants to put $x_i$ fraction of its ES in the market to share with the SFCs, where
\begin{eqnarray}
x_i\leq b_i = \left(s_i^\text{cap}-d_i\right).\label{eqn:1}
\end{eqnarray}
Here, $b_i$ is the maximum amount of battery space that the RU can share with the SFCs if the cost-benefit tradeoff for the sharing is attractive for it. $d_i$ is the amount of ES that the RU does not want to share, rather uses for its own needs, e.g., to run the essential loads in the future if there is any electricity disruption within the RU or if the price of electricity is very high.

To this end, to offer an ES space $x_i$, on the one hand, each RU $i$ decides on an reservation price $r_i$ per unit of energy. Hereinafter, we will use ES space and energy interchangeably to refer to the ES space that each RU might share with the SFCs. However, if the price $p_t$, which each RU received for sharing its ES, is lower than $r_i$, the RU $i$ removes its ES space $x_i$ from the market as the expected benefit from the joint sharing of ES is not economically attractive for it. On the other hand, each SFC $m\in\mathcal{M}$, that needs to share ES space with the RUs to store their energy, decides a reservation bid $a_m$, which represents the maximum unit price the SFC $m$ is willing to pay for sharing per unit of ES with the RUs in the smart community, to enter into the sharing market. And, if $a_m>p_t$, the SFC removes its commitment of joint ES ownership with RUs from the market due to the same reason as mentioned for the RU. A graphical representation of the concept of ES sharing and their decision making process of sharing the ES space of each RU $i$ with the SFCs are shown in Fig.~\ref{fig:batteryshare}. Please note that to keep the formulation simple, we do not include any specific storage model in the scheme. However, by suitably modeling some related parameters such as the storage capacity $s_i^\text{cap}$, and parameters like $d_i$ and $b_i$, the proposed scheme can be adopted for specific ES devices.
\begin{figure}[t]
\centering
\includegraphics[width=0.9\columnwidth]{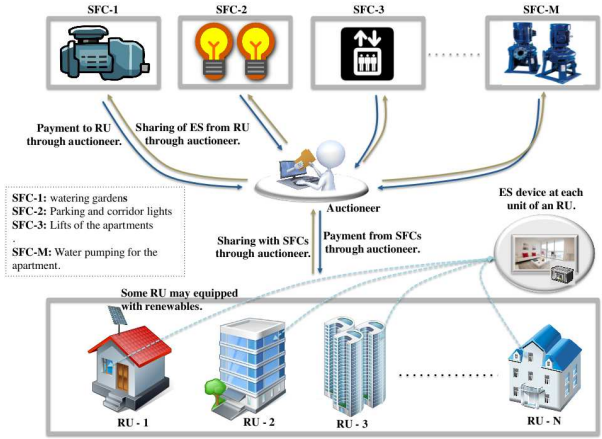}
\caption{Energy management in a smart community through auction process consisting of multiple RUs with ES devices, an auctioneer and a number of SFCs.} \label{fig:system_model}
\end{figure}

The interaction that arises from the choice of ES sharing price between the SFCs and RUs as well as the need of the SFCs to share the ES space to store their energy and the profits that the RUs can reap from allowing their ESs to be shared give rise to a market of ES sharing between the RUs and the SFCs in the smart grid. In this market, the involved $N$ RUs and $M$ SFCs will interact with each other to decide as to how many of them will take part in sharing the ESs between themselves, and also to agree on the ES sharing parameters such as the trading price $p_t$ and the amount of ES space to be shared. In the considered model, the RUs not only decide on the reservation prices $r_i$, but also on the amount of ES space $x_i$ that they are willing to share with the SFCs. The amount of $x_i$ is determined by the trade-off between between the economic benefits that the RU $i$ expects to obtain from giving the SFCs the joint ownership of its ES device and the associated reluctance $\alpha_i$ of the RU for such sharing. The reluctance to share ESs may arise from the RUs due to many factors. For instance, sharing would enable frequent charging and discharging of ESs that reduce the lifespan\footnote{Please note that the life time degradation due to charging and discharging may not true for all electromechanical systems such as redox-flow system.} of an ES device~\cite{Bradley-RSER:2009}. Hence, an RU $i$ may set its $\alpha_i$ higher so as to increase its reluctance to participate in the ES sharing. However, if the RU is more interested in earning revenue rather than increasing ES life time, it can reduce its $\alpha_i$ and thus get more net benefits from sharing its storage. Therefore, for a given set of bids $a_m,\forall m$ and storage requirement $q_m,\forall m$ by the SFCs, the maximum amount of ES $x_i$ that each RU $i$ will decide to put for sharing is strongly affected by the trading price $p_t$ and the reluctance parameter\footnote{Reluctance parameter refers to the opposite of preference parameter~\cite{Wayes-J-TSG:2012}.} $\alpha_i$ of each RU $i\in\mathcal{N}$ during the sharing process. In this context, we develop an auction based joint ES ownership scheme in the next section. We understand that the proposed scheme involves different types of users such as auctioneers, SFCs, and RUs. Therefore, the communication protocol used by them could be asynchronous. However, in our study we assume that the communication between different entities of the system are synchronous. This is mainly due to the fact that we assume our algorithm is executed once in a considered time slot, and the duration of this time slot can be one hour \cite{Derin:2010}. Therefore, synchronization is not a significant issue for the considered case and the communication complexity is affordable. For example, the auctioneer can wait for five minutes until it receives all the data from SFCs and the RUs and then the algorithm, which is proposed in Section~\ref{sec:algorithm}, can be executed.

\section{Auction Based ES Ownership}\label{sec:auction-ownership}
Vickrey auction is a type of sealed-bid auction scheme, where the bidders submit their written bids to the auctioneer without knowing the bids of others participating in the auction~\cite{Vickrey-JF:1961}. The highest bidder wins the auction but pays the \emph{second highest} bid price. Nevertheless, in this paper, we modify the classical Vickrey auction~\cite{Vickrey-JF:1961} to model the joint ES ownership scheme for a smart community consisting of multiple customers (i.e., the SFCs) and multiple owners of ES devices (i.e., the RUs). The modification is motivated by the following factors: 1) unlike the classical Vickrey auction, the modified scheme would enable the multiple owners and customers to decide simultaneously and independently whether to take part in the joint ES sharing through the determination rule of the proposed auction process, as we will see shortly; 2) the modification of the auction provides each participating RU $i$ with flexibility of choosing the amount of ES space that they may want to share with the SFCs in cases when the auction price\footnote{Hereinafter, $p_t$ will be used to refer to auction price instead of sharing or trading price.} $p_t$ is lower than their expected reservation price $r_i$; and 3) finally, the proposed auction scheme provides solutions that satisfy both the \emph{incentive compatibility} and \emph{individual rationality} properties, as we will see later, which are desirable in any mechanism that adopts auction theory~\cite{Saad-CSmartgridComm:2011}.

To this end, the proposed auction process, as shown in Fig.~\ref{fig:system_model}, consists of three elements:
\begin{enumerate}
\item Owner: The RUs in set $\mathcal{N}$, that own the ES devices, and expect to earn some economic benefits, e.g., through maximizing a utility function, by letting the SFCs to share some fraction of their ES spaces.
\item Customer: The SFCs in set $\mathcal{M}$, that are in need of ESs in order to store some excess electricity at a particular time of the day. The SFCs offer the RUs a price with a view to jointly own some fraction of their ES devices.
\item Auctioneer: A third party (e.g., estate or building manager), that controls the auction process between the owners and the customers according to some predefined rules.
\end{enumerate}
The proposed auction policies consist of A) determination rule, B) payment rule and C) storage allocation rule. Here, \emph{determination rule} allows the auctioneer to determine the maximum limit for the auction price $p_t^\text{max}$ and the number of SFCs and RUs that will actively take part in the ES sharing scheme once the auction process is initiated. The \emph{payment rule} enables the auctioneer to decide on the price that the customer needs to pay to the owners for sharing their ES devices, which allows the RUs to decide how much storage space they will be putting into the market to share with the SFCs. Finally, the auctioneer allocates the ES spaces for sharing for each SFC following the \emph{allocation rule} of the proposed auction. It is important to note that although both the customers and owners do not have any access to others private information such as the amount of ES to be shared by an RU or the required energy space by any SFC, the rules of auction are known to all the participants of the joint ownership process.

The proposed scheme initially determines the set of SFCs $\subset\mathcal{M}$ and RUs $\subset\mathcal{N}$ that will effectively take part in the auction mechanism once the upper bound of the auction price $p_t^\text{max}$ is determined. Eventually, the payment and the allocation rules are executed in the course of the auction plan.
\subsection{Determination Rule} The determination rule of the proposed scheme is executed by the following steps (inspired from~\cite{Huang-doubleauction:2002}):
\begin{enumerate}[i)]
\item The RUs of set $\mathcal{N}$, i.e., the owners of the ESs, declare their reservation price $r_i, \forall i$ in an increasing order, which we can consider, without loss of generality, as:
\begin{eqnarray}
r_1<r_2<\hdots<r_N.\label{eqn:ordered-rp}
\end{eqnarray}
The RUs submit the reservation price along with the amount $x_i$ of ES that they are interested to share with the SFCs to the auctioneer.
\item The SFCs' bidding prices, i.e., $a_m,\forall m$, are arranged in a decreasing order, i.e.,
\begin{eqnarray}
a_1>a_2>\hdots>a_M.\label{eqn:ordered-rb}
\end{eqnarray}
The SFCs submit to the auctioneer along with the quantity $q_m~\forall m$ of ES that they require.
\item Once the auctioneer receives the ordered information from the RUs and the SFCs, it generates the aggregated supply (reservation price of the RUs versus the amount of ES the RUs interested to share) and demand curves (reservation bids $a_m$ verses the quantity of ES $q_m$ needed) using \eqref{eqn:ordered-rp} and \eqref{eqn:ordered-rb} respectively.
\item The auctioneer determines the number of of participating SFCs $K$ and RUs $J$ that satisfies $a_K\geq r_J$ from the intersection of the two curves using any standard numerical method \cite{Huang-doubleauction:2002}.
\end{enumerate}
\begin{figure}[t!]
\centering
\includegraphics[width=0.9\columnwidth]{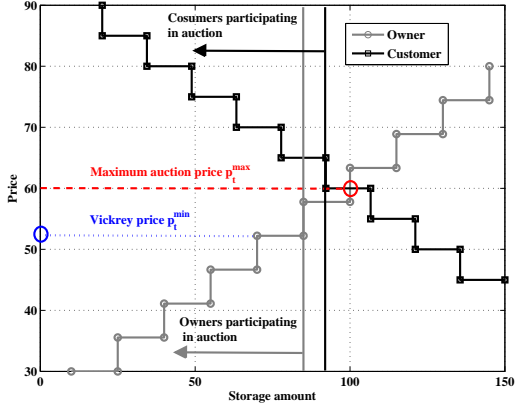}
\caption{Determination of the Vickrey price, the maximum auction price, and the number of participating RUs and SFCs in the auction process.} \label{fig:doubleauction}
\end{figure}

As soon as the SFC $K\leq M$ and RU $J\leq N$ are determined from the intersection point, as shown in Fig.~\ref{fig:doubleauction}, an important aspect of the auction mechanism is to determine the number of SFCs and RUs, which will take part in the joint ownership of ESs. We note that once the number of SFCs $K$ and RUs $J$ are determined, the following relationship holds for the rest of the SFCs and RUs in the network:
\begin{eqnarray}
 a_m<r_i;~\forall m\in\mathcal{M}/\{1,2,\hdots,K\},~\forall i\in\mathcal{N}/\{1,2,\hdots,J\}.
\end{eqnarray}
Hence, the joint ownership of ES would be a detrimental choice for the RUs and the SFCs within the set $\mathcal{N}/\{1,2,\hdots,J\}$ and $\mathcal{M}/\{1,2,\hdots,K\}$ respectively, which consequently remove them from the proposed auction process. Now, one desirable property of any auction mechanism is that no participating agents in the auction mechanism will cheat once the payment and allocation rules are being established. To this end, we propose that, once $J$ and $K$ are determined, $K-1$  SFCs and $J-1$ RUs will be engaged in the joint ES sharing process, which is a necessary condition for matching total demand and supply while maintaining a truthful auction scheme~\cite{Huang-doubleauction:2002}. Nevertheless, if truthful auction is not a necessity, SFC $K$ and RU $J$ can also be allowed to participate in the joint ES ownership auction.

\subsection{Payment Rule}\label{sec:payment} We note that the intersection of the demand and supply curves demonstrates the highest reservation price $p_t^\text{max}$ for the participating $J-1$ RUs. According to the \emph{Vickrey} auction mechanism~\cite{Vickrey-JF:1961}, the auction price for sharing the ES devices would be the second highest reservation price, i.e., the Vickrey price, which will be indicated as $p_t^\text{min}$ hereinafter. However, we note that this second highest price might not be considerably beneficial for all the participating RUs in the auction scheme. In contrast, if $p_t$ is set to $p_t = p_t^\text{max}$, the price could be detrimental for some of the SFCs. Therefore, to make the auction scheme attractive and beneficial to all the participating RUs and, at the same time, to be cost effective for all the SFCs, we strike a balance between the $p_t^\text{max}$ and $p_t^\text{min}$. To do so, we propose a scheme for deciding on both the auction price $p_t$ and the amount of ESs $x_i$ that RUs will put into the market for sharing according to $p_t$. In particular, we propose a Stackelberg game between the auctioneer that decides on the auction price $p_t$ to maximize the average cost savings to the SFCs as well as satisfying their desirable needs of ESs, and the RUs, that decide on the vector of the amount of ES $\mathbf{x} = [x_1, x_2,\hdots, x_{J-1}]$ that they would like to put into the market for sharing such that their benefits are maximized. Please note that the solution of the proposed problem formulation can also be solved following other distributed algorithms, e.g., algorithms designed via the bi-level optimization technique~\cite{Oduguwa:2002}.

\emph{Stackelberg game}: Stackelberg game is a multi-level decision making process, in which the leader of the game takes the first step to choose its strategy. The followers, on the other hand, choose their strategy in response to the decision made by the leader. In the proposed game, we assume the auctioneer as the leader and the RUs as the followers. Hence, it can be seen as a single-leader-multiple-follower Stackelberg game (SLMFSG). We propose that the auctioneer, as a leader of the SLMFSG $\Gamma$, will take the first step to choose a suitable auction price $p_t$ from the range $[p_t^\text{min},p_t^\text{min}]$. Meanwhile, each RU $i\in\{1, 2, \hdots,J-1\}$, as a follower of the game, will play its best strategy by choosing a suitable $x_i\in\left[0,b_i\right]$ in response to the price $p_t$ offered by the auctioneer. The best response strategy of each RU $i$ will stem from a utility function $U_i$, which captures the benefit that an RU $i$ can gain from deciding on the amount of ES $x_i$ to be shared for the offered price. Whereas the auctioneer chooses the price $p_t$ with a view to maximize the average cost savings $Z$ of the SFCs in the network. Now to capture the interaction between the auctioneer and the RUs, we formally define the SLMFSG $\Gamma$ as
\begin{eqnarray}
\Gamma = \lbrace\{\{1,2,\hdots,J-1\},\{\text{Auctioneer}\}\}, \{U_i\}_{i\in\{1,2,\hdots,J-1\}},\nonumber\\ \{\mathbf{X}_i\}_{i\in\{1,2,\hdots,J-1\}},Z, p_t\rbrace,
\label{eqn:Gamma}
\end{eqnarray}
which consists of: i) the set of RUs $\{1,2,\hdots,J-1\}$ participating in the auction scheme and the auctioneer; ii) the utility $U_i$ that each RU $i$ reaps from choosing a suitable strategy $x_i$ in response to the price $p_t$ announce by the auctioneer; iii) the strategy set $\mathbf{X}_i$ of each RU $i\in\{1,2,\hdots,J-1\}$; iv) the \emph{average} cost savings $Z$ that incurred to each SFC $m\in\{1,2,\hdots,K-1\}$ from the strategy chosen by the auctioneer, and v) the strategy $p_t\in\left[p_t^\text{min}, p_t^\text{max}\right]$ of the auctioneer.

Now, the utility function $U_i$, which defines the benefits that an RU $i$ can attain from sharing $x_i$ amount of its ES with the SFCs, is proposed to be
\begin{eqnarray}
U_i(x_i) = (p_t - r_i)x_i - \alpha_i x_i^2,~x_i\leq b_i, \label{eqn:utility}
\end{eqnarray}
where, $\alpha_i$ is the reluctant parameter of RU $i$, and $r_i$ is the reservation price set by RU $i$. $U_i$ mainly consists of two parts. The first part $(p_i - r_i)x_i$ is the utility in terms of its revenue that an RU $i$ obtains from sharing its $x_i$ portion of ES device. The second part $\alpha_i x_i^2$, on other hand, is the negative impact in terms of liability on the RU $i$ stemming from sharing its ES with the SFC. This is mainly due to the fact that once an RU decides to share its $x_i$ amount of storage space with an SFC, the RU can only use $s_i^\text{cap} - x_i$ amount of storage for its own use. The term $\alpha_i x_i^2$ captures this restriction of the RU on the usage of its own ES. In \eqref{eqn:utility}, the reluctance parameter $\alpha_i$ is introduced as a \emph{design parameter} to measure the degree of unwillingness of an RU to take part in energy sharing. In particular, a higher value of $\alpha_i$ refers to the case when an RU $i$ is more reluctant to take part in the ES sharing, and thus, as can be seen from \eqref{eqn:utility}, even with the same ES sharing attains a lower net benefit. Thus, $U_i$ can be seen as a net benefit to RU $i$ for sharing its ES. The utility function is based on the assumption of a non-decreasing marginal utility, which is suitable for modeling the benefits of power consumers, as explained in~\cite{Samadi-C-Smartgridcomm:2010}. In addition, the proposed utility function also possesses the following properties: i)  the utility of any RU increases as the amount of price $p_t$ paid to it for sharing per unit of ES increases; ii) as the reluctance parameter $\alpha_i$ increases, the  RU $i$ becomes more reluctant to share its ES, and consequently the utility decreases; and iii) for a particular price $p_t$, the more an RU shares with the SFCs, the less interested it becomes to share more for the joint ownership. To that end, for a particular price $p_t$ and reluctance parameter $\alpha_i$, the objective of RU $i$ is
\begin{eqnarray}
\max_{x_i}\left[ (p_t - r_i)x_i - \alpha_i x_i^2\right],~x_i\leq b_i.\label{eqn:obj-RU}
\end{eqnarray}
In the proposed approach, each RU $i$ iteratively responses to the strategy $p_t$ chosen by the auctioneer independent of other RUs in set $\{1, 2, \hdots, J-1\}/\{i\}$. The response of $i$ is affected by the offered price $p_t$, its reluctance parameter $\alpha_i$ and the initial reservation price $r_i$.

However, we note that the auctioneer does not have any control over the decision making process of the RUs. It only sets the auction price $p_t$ with a view to maximize the cost savings $Z$, with respect to the cost with the initial bidding price, for the SFCs. To this end, the target of auctioneer is assumed to maximize the average cost savings
\begin{eqnarray}
Z=\left(\frac{\sum_{m=1}^{K-1}(a_m-p_t)}{K-1}\right)\sum_{i=1}^{J-1} x_i\label{eqn:cost-savings}
\end{eqnarray}
by choosing an appropriate price $p_t$ to offer to each RU from the range $[p_t^\text{min}, p_t^\text{max}]$. Here, $\frac{\sum_{m=1}^{K-1}(a_m-p_t)}{K-1}$ is the average saving in auction price that the SFCs pay to the RUs for sharing the ESs and $\sum_i x_i$ is the total amount of ES that all the SFCs share from the RUs. From $Z$, we note that the cost savings will be more if $p_t$ is lower for all $m\in\{1,2,\hdots,K-1\}$. However, this is conflicted by that fact that a lower $p_t$ may lead to the choice of lower $x_i~\forall i\in\{1,2,\hdots,J-1\}$ by the RUs, which in turn will affect the cost to the SFCs. Hence, to reach a desirable solution set $(\mathbf{x}^*,p_t^*)$, the auctioneer and the RUs continue to interact with each other until the game reaches a Stackelberg equilibrium (SE).
\begin{definition}
Let us consider the game $\Gamma$ as described in \eqref{eqn:Gamma}, where the utility of each RU $i$ and the average utility per SFC are described via $U_i$ and $Z$ respectively. Now, $\Gamma$ will reach a SE $(\mathbf{x}^*,p_t^*)$, if and only if the solution of the game satisfies the following set of conditions:
\begin{eqnarray}
U_i(x_i^*, \mathbf{x}_{-i}^*,p_t^*)\geq U_i(x_i, \mathbf{x}_{-i}^*,p_t^*),\label{eqn:utilitymax}~\forall i\in\{1,2,\hdots,J-1\},\nonumber\\\forall x_i\in\mathbf{X}_i, p_t^*\in[p_t^\text{min}, p_t^\text{max}],
\end{eqnarray}
and
\begin{eqnarray}
\frac{\sum_{m=1}^{K-1}(a_m-p_t^*)}{K-1}\sum_i x_i^*\geq \frac{\sum_{m=1}^{K-1}(a_m-p_t)}{K-1}\sum_i x_i^*,
\label{eqn:costsaving}
\end{eqnarray}
where $\mathbf{x}_{-i} = \left[x_1, x_2,\hdots, x_{i-1}, x_{i+1}, \hdots, x_{J-1}\right]$.
\end{definition}
Hence, according to \eqref{eqn:utilitymax} and \eqref{eqn:costsaving}, both the RUs and the SFCs achieve their best possible outcomes at the SE. Hence, neither the RUs nor the auctioneer will have any incentive to change their strategies as soon as the game $\Gamma$ reaches the SE. However, achieving an equilibrium solution in pure strategies is not always guaranteed in non-cooperative games~\cite{Wayes-J-TSG:2012}. Therefore, we need to investigate whether the proposed $\Gamma$ possesses an SE or not.
\begin{theorem}
There always exists a unique SE solution for the proposed SLMFSG $\Gamma$ between the auctioneer and the participating RUs in set $\{1,2,\hdots,J-1\}$.
\label{theorem:1}
\end{theorem}
\begin{proof}
Firstly, we note that the strategy set of the auctioneer is non-empty and continuous within the range $\left[p_t^\text{min}, p_t^\text{max}\right]$. Hence, there will always be a non-empty strategy for the auctioneer that will enable the RUs to offer some part of their ES, within their limits, to the SFCs. Secondly, for any price $p_t$, the utility function $U_i$ in \eqref{eqn:utility} is strictly concave with respect of $x_i,~\forall i\in\{1, 2, \hdots, J-1\}$, i.e., $\frac{\delta^2 U_i}{\delta x_i^2}<0$. Hence, for any price $p_t\in\left[p_t^\text{min}, p_t^\text{max}\right]$, each RU will have a unique $x_i$, which will be chosen from a bounded range $[0, b_i]$ and maximize $U_i$. Therefore, it is evident that as soon as the scheme will find a unique $p_t^*$ such that the average utility $Z$ per SFC attains a maximum value, the SLMFSG $\Gamma$ will consequently reach its unique SE.

To this end, first we note that the amount of ES $x_i^*$, at which the RU $i$ achieves its maximum utility in response to a price $p_t$ can be obtained from \eqref{eqn:utility},
\begin{equation}
x_i^* = \frac{p_t - r_i}{2\alpha_i}.\label{eqn:SE-energy}
\end{equation}
Now, replacing the value of $x_i^*$ in \eqref{eqn:cost-savings} and doing some simple arithmetics, the auction price $p_t^*$, which maximizes the average cost savings to the SFCs  can be found as
\begin{equation}
p_t^* = \frac{\left(\sum_{i=1}^{J-1}\frac{1}{2\alpha_i}\right)\left(\sum_{m=1}^{K-1}a_m\right) + \sum_{i=1}^{J-1}\frac{r_i(K-1)}{2\alpha_i}}{\sum_{i=1}^{J-1}\frac{K-1}{\alpha_i}},
\label{eqn:SE-price}
\end{equation}
where $a_m$ for any $m\in\{1, 2, \hdots, K-1\}$ and $\alpha_i$ for any $i\in\{1, 2, \hdots, J-1\}$ is exclusive. Therefore, $p_t^*$ is unique for $\Gamma$, and thus Theorem~\ref{theorem:1} is proved.
\end{proof}
\subsection {Algorithm for payment}\label{sec:algorithm}
\hspace{1mm}To attain the SE, the auctioneer, which has the information of $a_m, m = \{1, 2, 3, \hdots, K-1\}$, needs to communicate with each RU. It is considered that the auctioneer does not have any knowledge of the private information of the RUs such as $\alpha_i,~\forall i$. In this regard, in order to decide on a suitable auction price $p_t$ that will be beneficial for both the RUs and the SFCs, the auctioneer and the RUs interact with one another. To capture this interaction, we design an iterative algorithm, which can be implemented by the auctioneer and the RUs in a distributed fashion to reach the unique SE of the proposed SLMFSG. The algorithm initiates with the auctioneer who sets the auction price $p_t$ to $p_t^\text{min}$ and the optimal average cost saving per SFC $Z^*$ to $0$. Now, in each iteration, after having the information on the offered auction price by the auctioneer, each RU $i$ plays its best response $x_i\leq b_i$ and submits its choice to the auctioneer. The auctioneer, on other hand, receives the information on $\mathbf{x} = [x_1, x_2, \hdots, x_{J-1}]$ from all the participating RUs and determines the average cost savings per SFC $Z$ from its knowledge on the reservation bids $[a_1, a_2, \hdots, a_{K-1}]$ and using \eqref{eqn:cost-savings}. Then, the auctioneer compares the $Z$ with $Z^*$. If $Z>Z^*$, the auctioneer updates the optimal auction price to the one recently offered and sends a new choice of price to the RUs in the next iteration. However, if $Z\leq Z^*$, the auctioneer keeps the same price and offers another new price to the RUs in the next iteration. The iteration process continues until the conditions in \eqref{eqn:utilitymax} and \eqref{eqn:costsaving} are satisfied, and hence the SLMFSG reaches the SE. We show the step-by-step process of the proposed algorithm in Algorithm~\ref{algorithm:1}.

\begin{algorithm}[h]
\caption{Algorithm for SLMFSG to reach the SE}
\label{algorithm:1}
\begin{algorithmic}[1]
\small
\STATE Initialization: $p_t^*=p_t^\text{min}$, $Z^*=0$.
\FOR {Auction price $p_t$  from $p_t^\text{min}$ to $p_t^\text{max}$ }
   \FOR {Each RU $i \in \{1,2,\hdots,(J-1)\}$}
        \STATE RU $i$ adjusts its amount of ES $x_i$ to share according to
        \begin{equation}\label{eqn:alg-1}
           x_i^* = {\rm{arg}}{\kern 1pt} {\kern 1pt} {\kern 1pt} \mathop {\max }\limits_{0 \le {x_i} \le b_i} {\kern 1pt} {\kern 1pt} {\kern 1pt} [(p_t - r_i)x_i - \alpha_i x_i^2].
        \end{equation}
   \ENDFOR
    \STATE The auctioneer computes the average cost savings to SFCs
        \begin{equation}\label{eqn:alg-2}
    {Z} = \left(\frac{\sum_{m=1}^{K-1}(a_m-p_t)}{K-1}\right)\sum_{i=1}^{J-1} x_i^*.
        \end{equation} \\
     \IF {$Z\ge Z^*$}
         \STATE The auctioneer record the desirable price and maximum average cost savings
         \begin{equation}\label{eqn:alg-3}
           p_t^* = p_t, Z^* = Z.
         \end{equation}
     \ENDIF
\ENDFOR\\
\textbf{The SE $(\mathbf{x}^*, p_t^*)$ is achieved.}
\end{algorithmic}
\end{algorithm}
\begin{theorem}
The algorithm proposed in Algorithm~\ref{algorithm:1} is always guaranteed to reach SE of the proposed SLMFSG $\Gamma$.
\label{theorem:2}
\end{theorem}
\begin{proof}
In the proposed algorithm, we note that the choice of strategies by the RUs emanate from the choice $p_t$ of the auctioneer, which as shown in \eqref{eqn:SE-price} will always attain a non-empty single value $p_t^*$ at the SE due to its bounded strategy set $[p_t^\text{min},p_t^\text{max}]$. On the other hand, as the Algorithm~\ref{algorithm:1} is designed, in response to the $p_t^*$, each RU $i$ will choose its strategy $x_i$ from the bounded range $[0,b_i]$ in order to maximize its utility function $U_i$. To that end, due to the bounded strategy set and continuity of $U_i$ with respect to $x_i$, it is confirmed that each RU $i$ will always reach a fixed point $x_i^*$ for the given $p_t^*$. Therefore, the proposed Algorithm~\ref{algorithm:1} is always guaranteed to reach the unique SE of the SLMFSG.
\end{proof}
\subsection{Allocation Rule}\label{sec:allocation} Now, once the the amount of ES $x_i^*$ that each RU $i\in\{1,2,\hdots,J-1\}$ decides to put into the market for sharing in response to the auction price $p_t^*$ is determined, the auctioneer allocates the quantity $Q_i$ to be jointly shared by each RU $i$ and the SFCs according to following rule \cite{Huang-doubleauction:2002}:
\begin{eqnarray}
Q_i(\mathbf{x}) = \begin{cases}
x_i^* & \text{if $\sum_{i=1}^{J-1} x_i^*\leq\sum_{m=1}^{K-1} q_m$},\\
(x_i^* - \eta_i)^{+} & \text{if $\sum_{i=1}^{J-1}x_i^*>\sum_{m=1}^{K-1} q_m$},
\end{cases}
\label{eqn:allocation-rule}
\end{eqnarray}
where $(f)^{+} = \max(0,f)$ and $\eta_i$ is the allotment of the excess ES $\sum_{i=1}^{J-1}x_i^*-\sum_{m=1}^{K-1}q_m$ that an RU $i$ must endure. Essentially, the rule in \eqref{eqn:allocation-rule} emphasizes that if the requirements of the SFCs exceed the available ES space from the RUs, each RU $i$ will allow the SFCs to share all of the ES $x_i$ that it put into the market. However, if the available ES exceeds the total demand by the SFCs, then each RU $i$ will have to share a fraction of the oversupply $\sum_{i=1}^{J-1}x_i^* - \sum_{m=1}^{K-1}q_m$. Nonethless, this burden, if there is any, can be distributed in different ways among the participating RUs. For instance, the burden can be distributed either proportionally to the amount of ES $x_i^*$ that each RU $i$ shared with the SFCs or proportionally to the reservation price\footnote{Please note that the reservation price $r_i$ indicates how much each RU $i$ wants to be paid for sharing its ES with the SFCs, and thus affects the determination of total $\sum x_i^*$ and the total burden.} $r_i$ of each RU. Alternatively, the total burden can also be shared equally by the RUs in the auction scheme~\cite{Huang-doubleauction:2002}.

\subsubsection{Proportional allocation}\label{sec:Proportional allocation}In proportional allocation~\cite{Guojon:2012}, a fraction of the total burden $\eta_i$ is allocated to each RU $i$ in proportion to the reservation price $r_i$ (or, $x_i^*$) such that $\sum_i \eta_i = \sum_{i=1}^{J-1}x_i^* - \sum_{m=1}^{K-1}q_m$, which can be implemented as follows:
\begin{eqnarray}
\eta_i = \frac{\left(\sum_{i=1}^{J-1}x_i^* - \sum_{m=1}^{K-1}q_m\right)r_i}{\sum_i r_i},~i = [1, 2, \hdots, J-1].
\label{eqn:burden-proportion}
\end{eqnarray}
By replacing $r_i$ with $x_i^*$ in \eqref{eqn:burden-proportion}, the burden allocation can be determined in proportion to the shared ES by each RU.
\subsubsection{Equal allocation}\label{sec:equal-allocation}
According to equal allocation~\cite{Huang-doubleauction:2002}, each RU bears an equal burden
\begin{equation}
\eta_i = \frac{1}{J-1}\left(\sum_{i=1}^{J-1}x_i^* - \sum_{m=1}^{K-1}q_m\right), ~i = [1, 2, \hdots, J-1]
\label{eqn:burden}
\end{equation}
of the oversupply. 

Here it is important to note that, although proportional allocation allows the distribution of oversupply according to some properties of the RUs, equal allocation is more suitable to make the auction scheme strategy proof~\cite{Huang-doubleauction:2002}. Strategy proofness  is important for designing auction mechanisms as it encourages the participating players not to lie about their private information such as reservation price \cite{Saad-CSmartgridComm:2011}, which is essential for the acceptability and sustainability of such mechanisms in energy markets. Therefore, we will use equal allocation of \eqref{eqn:burden} for the rest of the paper.

\subsection{Properties of the Auction Process}We note that once the auction process is executed, there is always a possibility that the owners of the ES might cheat on the amount of storage that they wanted put into the market during auction~\cite{Ma-JTSG:2014}. In this context, we need to investigate whether the proposed scheme is beneficial enough, i.e., \emph{individually rational}, for the RUs such that they are not motivated to cheat, i.e., \emph{incentive compatible}, once the auction is executed.

Now for the \emph{individual rationality} property, first we note that all the players, i.e., the RUs and the auctioneer on behalf of the SFCs, take part in the SLMFSG to maximize their benefits in terms of their respected utility from their choice of strategies. The choice of the RUs is to determine vector of ES $\mathbf{x}^*$ such that each of the RU can be benefitted at its maximum. On the other hand,  the strategy of the auctioneer is to choose a price $p_t$ to maximize the savings of the SFCs. Accordingly, once both the RUs and the auctioneer reach such a point of the game when neither the owners nor the customers can be benefitted more from choosing another strategy, the SLMFSG reaches the SE. To this end, it is already proven in Theorem~\ref{theorem:1} that the proposed $\Gamma$ in this auction process must possesses a unique SE. Therefore, as a subsequent outcome of the Theorem~\ref{theorem:1}, it is clear that all the participants in the proposed auction scheme are individually rational, which leads to the following Corollary~\ref{corollary:1}.
\begin{corollary}
The proposed auction technique possesses the individual rationality property, in which the $J-1$ rational owners and $K-1$ rational customers actively participate in the mechanism to gain the higher utility.
\label{corollary:1}
\end{corollary}
\begin{theorem}
The proposed auction mechanism is incentive compatible, i.e., truthful auction is the best strategy for any RU $i\in\mathcal\{1,2,\hdots,J-1\}$ and SFC $m\in\mathcal\{1,2,\hdots,K-1\}$.
\label{theorem:3}
\end{theorem}
\begin{proof}
To validate Theorem \ref{theorem:3}, first we note that the choice of strategies by the RUs always guaranteed to converge to a unique SE, i.e., $\mathbf{x^*} = [x_1^*, x_2^*, \hdots, x_{J-1}^*]$ as proven in Theorem~\ref{theorem:1} and Theorem \ref{theorem:2}, which confirms the stability of their selections. Now, according to \cite{Huang-doubleauction:2002}, once the owners of an auction process, i.e., the RUs in this proposed case, decide on a stable amount of commodity, i.e., $x_i^*~\forall i\in\{1,2,\hdots,J-1\}$, to supply to or to share with the customers, the auction process always converges to a \emph{strategy-proof} auction if the allocation of commodity is conducted according to the rules described in \eqref{eqn:allocation-rule} and \eqref{eqn:burden}.  Therefore, neither any RU nor any SFC will have any intention to falsify their allocation once they adopt \eqref{eqn:allocation-rule} and \eqref{eqn:burden}~\cite{Huang-doubleauction:2002} for sharing the storage space of the RUs from their SE amount. Therefore, the auction process is \emph{incentive compatible}, and thus Theorem~\ref{theorem:3} is proved.
\end{proof}

\subsection{Adaptation to Time-Varying Case}\label{sec:time-varying}
To extend the proposed scheme to a time-varying case, we assume that the ES sharing scheme works in a time-slotted fashion where each time slot has a suitable time duration based on the type of application, e.g., $1$ hour \cite{Derin:2010}. It is considered that in each time slot all the RUs and SFCs take part in the proposed ES sharing scheme to decide on the parameters such as the auction price and the amount of ESs that needs to be shared. However, in a time-varying case, the amount of ES that an RU shares at time slot $t$ may be affected by the burden that the RU needed to bear in the previous time slot $t-1$. To this end, first we note that once the number of participating RUs and SFCs is decided for a particular time slot via the determination rule, the rest of the procedures, i.e., the payment and allocation rules are executed following the descriptions in Section \ref{sec:payment} and \ref{sec:allocation} respectively for the respective time slot. Now, if the total number of RUs and SFCs is fixed, the RUs and SFCs that participate in the modified auction scheme in any time slot is determined by their respective reservation and bidding prices for that time slot. Further, the proposed auction process may evolve across different time slots based on the change of the amount of ES that each participating RU $i$ may want to share and the change in the total amount of ES required for the SFCs in different time slots. Now, before discussing how the proposed modified auction scheme can be extended to a time-varying environment\footnote{Certain loads such as lifts and water pumps in large apartment buildings are not easy to schedule as they are shared by different users of the buildings. Hence, we focus on the time variation of the storage sharing process by the RUs of the considered system.}, first we define the following parameters:\\
$t$: index of time slot.\\
$T$: total number of time slot.\\
$r_{i,t}$: the reservation price of RU $i\in\mathcal{N}$ at time slot $t$.\\
$\mathbf{r}_i = [r_{i,1}, r_{i,2}, \hdots, r_{i,T}]$: is the reservation price vector for RU $i\in\mathcal{N}$.\\
$x_{i,t}$: the fraction of ES space that the RU $i$ wants to shares with the SFCs at time slot $t$.\\
$\mathbf{x}_i = [x_{i,1}, x_{i,2},\hdots, x_{i,T}]$: the vector of ES space shared by RU $i$ with the SFC during the total considered times.\\
$b_{i,t}$: maximum available ES of RU $i$ for sharing at time slot $t$.\\
$a_{m,t}$: the bidding price of each SFC $m\in\mathcal{M}$ at time slot $t$.\\
$\mathbf{a}_m = [a_{m,1}, a_{m,2}, \hdots, a_{m,T}]$: is the reservation price vector for SFC $m\in\mathcal{N}$.\\
$q_{m,t}$: the required ES space by each SFC $m$ at time slot $t$.\\
$p_{t,t}$: the auction price at time slot $t$.\\
$U_{i,t}$: the benefit that each RU $i$ achieves at time slot $t$.\\
$Z_t$: the average cost saving per SFC at time slot $t$.\\
$\eta_{i,t}$: the burden that is shared by each participating RU at time slot $t$.\\
$K_t$: number of participating SFCs in the modified auction scheme at time slot $t$.\\
$J_t$: number of participating RUs in the modified auction scheme at time slot $t$.\\
To this end, the utility function $U_{i,t}$ of each RU $i$ and the average cost savings $Z_t$ per SFC at time slot $t$ can be defined as
\begin{eqnarray}
U_{i,t}(x_{i,t}) = (p_{t,t}-r_{i,t})x_{i,t} - \alpha_i x_{i,t}^2,
\label{eqn:time-varying:1}
\end{eqnarray}
and
\begin{eqnarray}
Z_t = \left(\frac{\sum_{m=1}^{K_t-1}(a_{m,t} - p_{t,t})}{K_t - 1}\right)\sum_{i=1}^{J-1}x_{i,t}
\label{eqn:time-varying:2}
\end{eqnarray}
respectively\footnote{Please note that in each time slot $t$, \eqref{eqn:time-varying:1} and \eqref{eqn:time-varying:2} are related with each other in a similar manner as \eqref{eqn:obj-RU} and \eqref{eqn:cost-savings} are related for the static case. However, unlike the static case, the execution of the auction process in each time slot $t$ is affected by the value of parameters such as $x_{i,t}$ and $p_t$ for that particular time slot.}. 

Now, at time slot $t$, the determination rule of the proposed scheme determines the number of participating RUs and SFCs based on their reservation and bidding prices for that time slot. The number of participation is also motivated by the available ES space of each RU and the requirement of each SFC. However, unlike the static case, in a time-varying environment the offered ES space by an RU at time slot $t$ is influenced by its contribution to the auction process in the previous time slot. For instance, if an RU $i$ receives a burden $\eta_{i,t-1}$ in time slot $t-1$, its willingness to share ES space $x_{i,t}$ at time slot $t$ may reduce. $x_{i,t}$ is also affected by the  maximum amount of ES $b_{i,t}$ available to RU $i$ at $t$. For simplicity, we assume that $b_{i,t}$ and $\alpha_{i,t}$ do not change over different time slots. Therefore, an RU $i$ can offer to share the same amount of ES space $x_{i,t}$ to the SFCs at time slot $t$ \emph{if it did not share} any amount in time slot $t-1$. An analogous example of such arrangement can be found in FIT scheme with ES device in which households are equipped with a dedicated battery to sell the stored electricity to the grid \cite{Goran:2011}. Nonetheless, $x_{i,t}$ is also affected by the amount of burden $\eta_{i,t-1}$ that an RU needed to bear due to an oversupply of ES spaces, if there was any, in the previous time slot. To this end, the amount of ES space that an RU $i$ can offer to the SFCs at $t$ can be defined as
\begin{eqnarray}
x_{i,t} = \begin{cases}
x_{i,t-1} & \text{if $i\notin J_{{t-1}}$}\\
\max(b_{i,t} -  (x_{i,t-1} - \eta_{i,t-1}),0) & \text{otherwise}
\end{cases}.
\label{eqn:timeVary-x}
\end{eqnarray}
The SFC $m$, on the other hand, decides on the amount of ES $q_{m,t}$ that it needs to share from the RUs at $t$ based on the random requirement of the shared facilities at $t$, the available shared ES space $q_{m,t-1}$ from time slot $t-1$,  and the random generation of renewable energy sources, where appropriate. Hence,
\begin{eqnarray}
q_{m,t} = f(q_{m,t-1}, \text{renewables, facility requirement}).
\end{eqnarray}
Now, if we assume that the fraction of shared ES available from previous time slot is negligible, i.e., $q_{m,t-1}\approx 0$, the requirement $q_{m,t}$ can be assumed to be random for each time slot $t$ considering the random nature of both renewable generation and energy requirement of shared facilities. Note that this assumption is particularly valid if the SFC uses all its shared ESs from the previous time slot for meeting the demand of the shared facilities and cannot use them in considered time slot. Nonetheless, please note that this assumption does not imply that the inter-temporal relationship between the auction process across different time slots is non-existent. The auction process in one time slot still depends on other time slots due to the inter-temporal dependency of $x_{i,t}$ via \eqref{eqn:timeVary-x}.

To this end, for the modeled $x_{i,t}~\forall i\in\mathcal{N}$ and $q_{m,t}~\forall m\in\mathcal{M}$, the proposed modified auction scheme studied in Section~\ref{sec:auction-ownership} can be adopted in each time slot $t=1, 2, \hdots, T$ with a view to maximize \eqref{eqn:time-varying:1} and \eqref{eqn:time-varying:2} $\forall t$. It is important to note that the reservation price vector $\mathbf{r}_i$ of each RU $i\in\mathcal{N}$  and the bidding price vector $\mathbf{a}_m$ of each SFC $m\in\mathcal{M}$ can be modeled through any existing time-varying pricing schemes such as time-of-use price \cite{Fang-J-CST:2012}. Now, $\mathbf{p}_{t}^* =  [p_{t,1}^*, p_{t,2}^*,\hdots, p_{t,T}^*]$ and $\mathbf{x}^* = [\mathbf{x}_1^*, \mathbf{x}_2^*, \hdots, \mathbf{x}_N^*]$ constitute the solutions of the proposed modified auction scheme in a time-varying condition, if the $\mathbf{x}^*$ comprises the solution vector of all ES spaces shared by the participating RUs in each time slot $t = 1, 2, \hdots, T$ for the auction price vector $\mathbf{p}_t^*$. Further, all the auction rules adopted in each time slot of the proposed time-varying case will be similar to the rules discussed in Section~\ref{sec:auction-ownership}. Hence, the solution of the proposed modified auction scheme for a time-varying environment also possesses the incentive compatibility and individual rationality properties for each time slot.
\section{Case Study}\label{sec:case-study}
For numerical case studies, we consider a number of RUs at different blocks in a smart community that are interested in allowing the SFCs of the community to jointly share their ES devices. We stress that when there are a large number of RU and SFCs in the system, the reservation and bidding prices will vary significantly from one another. Therefore, it will be difficult to find an intersection point to determine the highest reservation price $p_t^\text{max}$ according to the determination rule. So, in this paper, we limit ourself to around $6-10$ RUs. However, having 6-10 RUs can in fact cover a large community, e.g., through aggregation such as discussed in \cite{Gkatzikis:2013,Wayes-J-TSG:2012}. Here, each RU is assumed to be a group of $[5,~25$] households, where each household is equipped with a battery of capacity $25$ kilo-Watt hour (kWh)~\cite{battery:2013}. The reluctance parameter of all RUs are assumed to be similar, which is taken from range of $[0, 0.1]$. It is important to note that $\alpha_i$ is considered as a design parameter in the proposed scheme, which we used to map the reluctance of each RU to share its ES with the SFCs. Such reluctance of sharing can be affected by parameters like ES capacity, the condition of the environment (if applicable) and the RU's own requirement. Now, considering the different system parameters in our proposed scheme, we capture these two extremes with 0 (not reluctant) and 0.1 (highly reluctant). The required electricity storage for each SFC is assumed to be within the range of $[100,~500]$ kWh. Nevertheless, the required ES for sharing could be different if the usage pattern by the users changes. Since, the type of ESs (and their associated cost) used by different RUs can vary significantly~\cite{storageType:2012}, the choices of reservation price to share their ESs with the SFCs can vary considerably as well. In this context, we consider that the reservation price  set by each RU and SFC is taken from a range of [20, 70]. It is important to note that all chosen parameter values are particular to this study only, and may vary according the availability and number of RUs, requirements of SFCs, trading policy, time of the day/year and the country.

\begin{figure}[t!]
\centering
\includegraphics[width=\columnwidth]{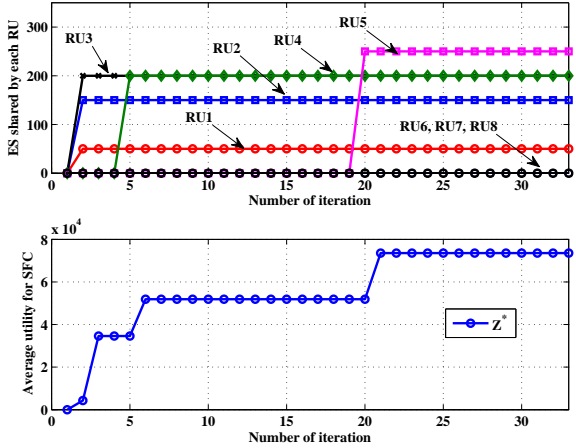}
\caption{Convergence of Algorithm~\ref{algorithm:1} to the SE. At SE, the average utility per SFC reaches its maximum and the ES that each RU wants to put into the market for share reaches a steady state level that maximize their benefits.} \label{fig:Convergence}
\end{figure}

Now, we first show the convergence of Algorithm~\ref{algorithm:1} to the SE of the SLMFG in Fig.~\ref{fig:Convergence}. For this case study, we assume that there are five SFCs in the smart grid community that are taking part in an auction process with eight RUs. From Fig.~\ref{fig:Convergence}, first we note that the proposed SLMFG reaches the SE after $20$ interations when the average cost savings per SFC reaches its maximum. Hence, the convergence speed, which is just few seconds, is reasonable. Nonetheless, an interesting property can be observed when we examine the choice of ES by each RU to put into the market for sharing. As can be seen from the figure, on the one hand, RU $1$, RU $2$, and RU $3$ reach the SE much quicker than RU $4$ and RU $5$. On the other hand, no interest for sharing any ES is observed for RU $4, 5$ and $6$. 

This is due to the fact that as the interaction between the auctioneer and the RUs continues, the auction price $p_t$ is updated in each iteration. In this regard, once the auction price for any RU becomes larger than its reservation price, it put all its reserve ES to the market with an intention to be shared by the SFCs. Due to this reason, RU $1$, RU $2$, and RU $3$ put their ESs in the market much sooner, i.e., after the $2^\text{nd}$ iteration, than RU $4$ and RU $5$ with higher reservation prices, whose interest for sharing ES reaches the SE once the auction price is encouraging enough for them to share their ESs after the $5^\text{th}$ and $20^\text{th}$ iterations. Unfortunately, the utilities of RU $6$, $7$, and $8$ are not convenient enough to take part in the auction process, and therefore their shared ES fractions are $0$. 

We note that the demonstration of the convergence of the SLMFSG to a unique SE subsequently demonstrates the proofs of Theorem~\ref{theorem:1}, Theorem~\ref{theorem:2}, Theorem~\ref{theorem:3} and Corollary~\ref{corollary:1}, which are strongly related to the SE as explained in the previous section. Now, we would like to investigate how the reluctance parameters of the RUs may affect their average utility from Algorithm~\ref{algorithm:1}, and thus affecting their decisions to share ES. To this end, we first determine the average utility that is experienced by each RU and SFC for a reluctance parameter of $\alpha_i = 0.001~\forall i$. Then considering the outcome as a benchmark, we show the effect of different reluctance parameters on the achieved average benefits of each SFC and RU in Table~\ref{table:1}. The demonstration of this property is necessary in order to better understand the working principle of the designed technique for ES sharing.

\begin{table}[t]
\centering
\caption{Change of average utility achieved by each SFC and each RU in the network (according to Algorithm~\ref{algorithm:1}) due to the change of the reluctance of each RU for sharing one kWh ES with the SFC.}
\includegraphics[width=\columnwidth]{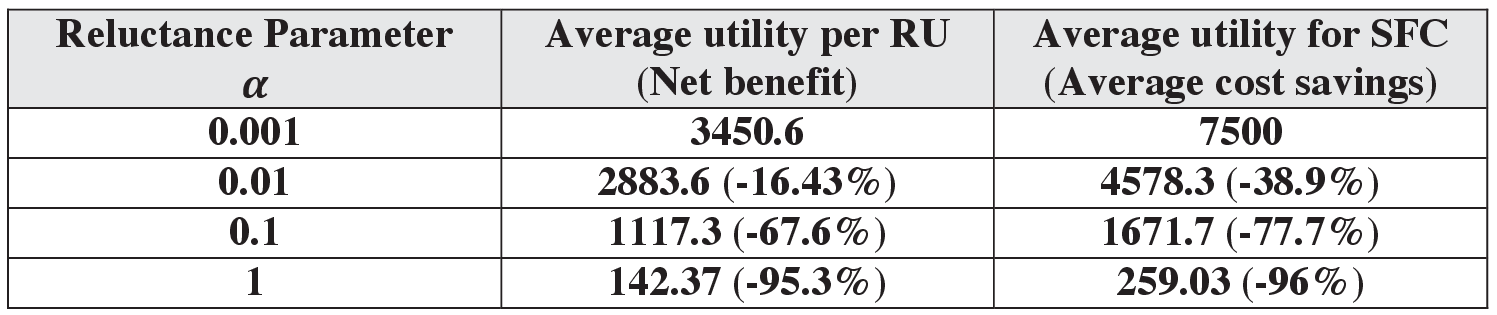}
\label{table:1}
\end{table}

According to Table~\ref{table:1}, as the reluctance of each RU increases, it becomes more uncomfortable, i.e., lower utility, to put its ES in the market to be jointly owned by the SFCs. As a consequence, it also affects the average utility achieved by each SFC. As shown in Table~\ref{table:1}, the reduction in average utilities per RU are $16.73\%$, $67.6\%$ and $95.3\%$ respectively compared to the average utility achieved by an RU at $\alpha_i = 0.001$ for every ten times reduction in the reluctance parameter. For similar settings, the reduction of average utility for the SFCs are $38.9\%$, $77.7\%$ and $96\%$ at $\alpha_i = 0.01, 0.1$ and $1$ respectively. Therefore, the proposed scheme will enable the RUs to put more storage in the auction  market if the related reluctance for this sharing is small. Note that although the current investment cost of batteries is very high compared to their relative short life times, it is expected that battery costs will go down in the near future \cite{Wang-JTSG:2013}  and become very popular for addressing intermittency of renewables \cite{Tesla:2015}. We have foreseen such a near future when our proposed scheme will be applicable to gain the benefit of storage sharing and thus motivate the RUs to keep their $\alpha_i~\forall i$ small. According to the observation from Table~\ref{table:1}, it can further be said that if the reluctance parameters of RUs change over either different days or different time slots, the performance of the system in terms of average utility per RU and average cost savings per SFC will change accordingly for the given system parameters.
\begin{figure}[t]
\centering
\includegraphics[width=\columnwidth]{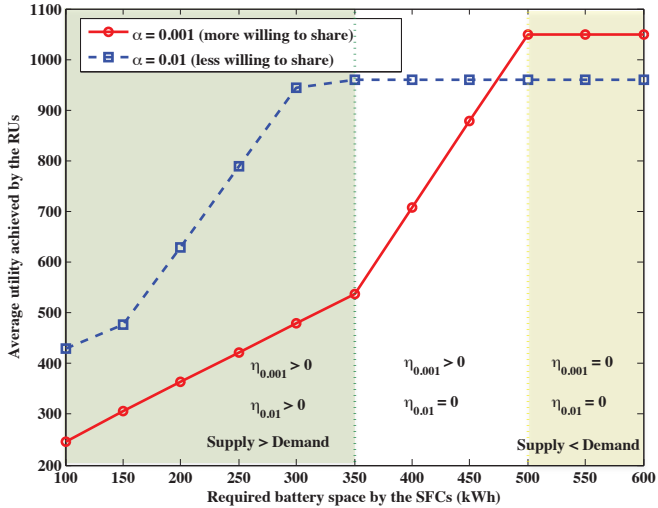}
\caption{Effect of change of required ES amount by the SFCs on the achieved average utility per RU.} \label{fig:UtilityVsBuyAmount}
\end{figure}

\begin{table*}[t]
\centering
\caption{Comparison of the change of average utility per RU in the smart grid system as the required total amount of energy storage required by the SFCs varies.}
\begin{tabular}{|c|c|c|c|c|c|c|}
\hline
Required ES space by the SFCs & 200 & 250 & 300 & 350 & 400 & 450\\
\hline
Average utility (net benefit) of RU for equal distribution (ED) scheme & 536.52 & 581.85 & 624.52 & 669.85 & 715.19 & 757.85\\
\hline
Average utility (net benefit) of RU for FIT scheme & 537.83 & 583.16 & 626.83 & 673.16 & 717.50 & 759.16\\
\hline
Average utility (net benefit) of RU for proposed scheme & 629.82 & 789.82 & 944.26 & 960.09 & 960.09 & 960.09\\
\hline
Percentage improvement (\%) compared to ED scheme &17.4 & 35.74 & 51.19 & 43.32 & 34.24 & 26.68\\
\hline
Percentage improvement (\%) compared to FIT scheme &17.1 & 35.43 & 50.63 & 42.61 & 33.81 & 26.46\\
\hline
\end{tabular}
\label{table:2}
\end{table*}

Once all the participating RUs put their ES amount into the auction market, they are distributed according to the allocation rule described in \eqref{eqn:allocation-rule} and \eqref{eqn:burden}. In this regard, we investigate how the average utility of each RU is altered as the total storage amount required by the SFCs changes from in the network. For this particular case, the considered total ES requirement of the SFCs is assumed to be $100, 150, 200, 250, 300, 350, 400, 450, 500, 550$ and $600$. In general, as shown in Fig.~\ref{fig:UtilityVsBuyAmount}, the average utility of each RU initially increases with the increase required by the SFCs and eventually becomes saturated to a stable value. This is due to the fact that as the required amount of ES increases, the RU can share more of its reserved ES that it put into the market with the SFCs with the determined auction price from the SLMFSG. Hence, its utility increases. However, each RU has a particular fixed ES amount that it puts into the market to share. Consequently, once the shared  ES amount reaches its maximum, even with the increase of requirement by the SFCs the RU cannot share more, i.e., $\eta_i = 0$. Therefore, its utility becomes stable without any further increment. Interestingly, the proposed scheme, as can be seen in Fig.~\ref{fig:UtilityVsBuyAmount}, favors the RUs with higher reluctance more when the ES requirement by the SFCs is relatively lower and favors the RUs with lower reluctance during higher demands. This is due to the way we have designed the proposed allocation scheme, which is dictated by the burden in \eqref{eqn:burden} and the allocation of ES through \eqref{eqn:allocation-rule}. We note that, according to \eqref{eqn:SE-energy}, if $\alpha_i$ is lower, the RU $i$ will put a higher amount of ES in the market to share. However, if the total required amount of ES is lower for the SFCs, it would put a higher burden on the RUs to carry. As a consequence, the relative utility from auction is lower. Nevertheless, if the requirement of the SFCs is higher, the sharing brings significant benefits to the RUs as can be seen from Fig.~\ref{fig:UtilityVsBuyAmount}. On the other hand, for higher reluctance, RUs tend to share a lower ES amount, which then enables them to endure a lower burden in case of lower demands from the SFCs. This consequently enhances their achieved utility. Nonetheless, if the requirement is higher from the SFCs, their utility reduces subsequently compared to the RUs with lower reluctance parameters. Thus, from observing the effects of different $\alpha_i$'s on the average utility per RU in Fig. \ref{fig:UtilityVsBuyAmount}, we understand that, if the total required ES is smaller, RUs with higher reluctance benefit more and vice versa. This illustrates the fact that even RUs with high unwillingness to share their ESs can be beneficial for SFCs of the system if their required ESs are small. However, for a higher requirement, SFCs would benefit more from having RUs with lower reluctances as they will be interested in sharing more to achieve higher average utilities.

Now, we discuss the computational complexity of the proposed scheme, which is greatly reduced by the determination rule of the modified auction scheme as this rule determines the actual number of participating RUs and SFCs in the auction. We also note that after determining the number of participating SFCs and RUs, the auctioneer iteratively interacts with each of the RUs and sets the auction price with a view to increase the average savings for the SFC. Therefore, the main computational complexity of the modified auction scheme stems from the interactions between the auctioneer and the participating RUs to decide on the auction price. In this context, the computational complexity of the problem falls within a category of that of a single leader multiple follower Stackelberg game, whose computational complexity, which can be approximated to increase linearly with the number of followers \cite{Wayes-J-TSG:2012}, and is shown to be reasonable in numerous studies such as in \cite{Tushar-TIE:2014} and \cite{Wayes-J-TSG:2012}. Hence, the computational complexity is feasible for adopting the proposed scheme.

Having an insight into the properties of the proposed auction scheme, we now demonstrate how the technique can benefit the RUs of the smart network compared to existing ES allocation schemes such as equal distribution (ED)~\cite{Wayes-J-TSG:2012} and FIT schemes \cite{Goran:2011}. ED is essentially an allocation scheme that allows the SFCs to meet their total storage requirements by sharing the total requirement equally from each of the participating RUs. We assume that if the shared ES amount exceeds the total amount of reservation storage that an RU puts into the market, the RU will share its full reservation amount. In FIT, which is a popular scheme for energy trading between consumers and the grid, we assume that each RU prefers to sell the same storage amount of energy to the grid at an FIT price rate, e.g., $22$ cents/kWh~\cite{LIPA} instead of sharing the same fraction of storage with the SFC. To this end, the resulting average utilities that each RU can achieve from sharing its ES space with the SFCs by adopting the proposed, ED, and FIT schemes are shown in Table~\ref{table:2}.

From Table~\ref{table:2}, first we note that as the amount of required ES by the SFCs increases the average utility achieved per RU also increases for all the cases. The reason for this increment is explained in Fig.~\ref{fig:UtilityVsBuyAmount}. Also, in all the studied cases, the proposed scheme shows a considerable performance improvement compared to the ED  and FIT schemes. An interesting trend of performance improvement can be observed if we compare the performance of the proposed scheme with the ED and FIT performances for each of the ES requirements. In particular, the performance of the proposed scheme is higher as the requirement of the ES increases from $200$ to $350$. However, the improvement is relatively less significant as the ES requirement switches from $400$ to $450$. This change in performance can be explained as follows:

In the proposed scheme, as we have seen in Fig.~\ref{fig:UtilityVsBuyAmount}, the amount of ES shared by each participating RU is influenced by their reluctance parameters. Hence, even the demand of the SFCs could be larger, the RUs may choose not to share more of their ES spaces once their reluctance is limited. In this regard, the RUs in the current case study increase their share of ES as the requirement by the  SFCs increases, which in turn produces higher revenue for the RUs. Furthermore, once the RUs choice of ESs reach the saturation, the increase in demand, i.e., from $200$ to $350$ in this case, does not affect their share. As a consequence, their performance improvement is not as noticeable as the previous four cases. Nonetheless, for all the considered cases, the auction process performs superior to the ED scheme with an average performance improvement of $34.76\%$, which clearly shows the value of the proposed methodology to adopt joint ES sharing in smart grid. The performance improvement with respect to the FIT scheme, which is $34.34\%$ on average, is due to the difference between the determined auction price and the price per unit of energy for the FIT scheme.
\begin{figure}[t!]
\centering
\includegraphics[width=0.9\columnwidth]{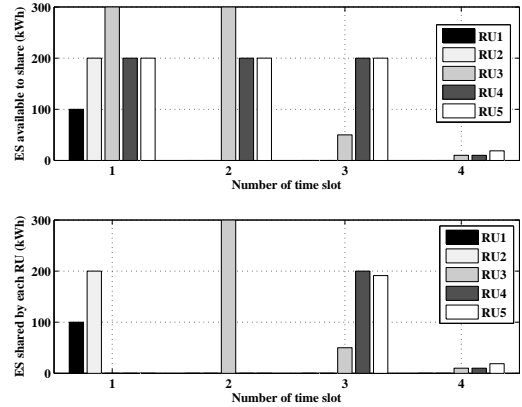}
\caption{Demonstration of how the proposed modified auction scheme can be extended to time varying system. The reservation ES amount varies by the RUs varies between different time slots based on their sharing amount in the previous time slot. The total required storage by the SFCs is chosen randomly due to the reasons explained in Section~\ref{sec:time-varying}.} \label{fig:time-varying}
\end{figure}

Finally, we show how the decision making process of each RU in the system is affected by its decision in the previous time slot and the total storage requirement by the SFCs. The total number of time slots that are considered to show this performance analysis is four. In this context, we assume that there are five RUs in the system with ES of $100, 200, 300, 200$ and $200$ kWh respectively to share with the SFCs. The total ES requirements of the SFCs for considered four time slot are $500, 250, 500$, and 100. Please note that these numbers are considered for this case study only and may have different values for different scenarios. Now, in Fig.~\ref{fig:time-varying}, we show the available ES to each of the RUs at the begining of each time slot and how much they are going to share if the modified auction scheme is adopted in each time slot. For a simple analysis, we assume that once an RU shares its total available ES, it cannot share its ES for the remaining of the time slots. The reservation prices are considered to change from one time to the next based on a predefined time of use price scheme. Now, as can be seen from Fig.~\ref{fig:time-varying}, in time slot $1$, RU1 and RU2 share all their available ESs with the SFC, whereby other RUs do not share their ESs due to the reasons explained in Fig.~\ref{fig:Convergence}. Since, the total requirement is $500$, therefore neither of RU1 and RU2 needs to carry any burden. In time slot $3$, only RU3 shares its ESs of $300$ to meet the requirement. As the SFC's requirement is lower than the supply, RU3 needs to carry a burden of $50$ kWh. Similarly, in time slot $3$ and $4$, all of RU3, RU4 and RU5 take part in the energy auction scheme as they have enough ES to share with the SFC. However, the ES to share in time slot $4$ stems from the burden of oversupply from time slot $3$. The scheme is not shown for more than time slot $4$ as the available ES from all RUs is already shared by the SFCs by the end of time slot $4$. Thus, the proposed modified auction scheme can successfully capture the time variation if the scheme is modified as given in Section~\ref{sec:time-varying}.

\section{Conclusion}\label{sec:conclusion}
In this paper, we have modeled a modified auction based joint energy storage ownership scheme between a number of residential units (RUs) and shared facility controllers (SFCs) in smart grid. We have designed a system and discussed the determination, payment and allocation rule of the auction, where the payment rule of this scheme is facilitated by a Single-leader-multiple-follower Stackelberg game (SLMFSG) between the auctioneer and the RUs. The properties of the auction scheme and the SLMFSG have been studied, and it has been shown that the proposed auction possesses the \emph{individual rationality} and the \emph{incentive compatibility} properties leveraged by the unique Stackeberg equilibrium of the SLMFSG. We have proposed an algorithm for the SLMFSG, which has been shown to be guaranteed to reach the SE and that also facilitates the auctioneer and the RUs to decide on the auction price as well as the amount of ES to be put into the market for joint ownership. 

A compelling extension of the proposed scheme would be to study of the feasibility of scheduling of loads such as lifts and water machines in shared space. Another  interesting research direction would be to determine how a very large number of SFCs or RUs with different reservation and bidding prices can take part in such a modified auction scheme. One potential way to look at this problem can be from a cooperative game-theoretic point-of-view in which the SFCs and RUs may cooperate to decide on the amount of reservation ES and bidding price they would like to put into the market so as to participate in the auction and benefit from sharing. Another very important, yet interesting, extension of this work would be to investigate how to quantify the reluctance of each RU to participate in the ES sharing. Such quantification of reluctance (or, convenience) will also enable the practical deployment of many energy management schemes already described in the literature.

\end{document}